\documentclass{article}

\usepackage{graphicx}      
\usepackage{natbib}        
 
\usepackage{xcolor}
 
\usepackage{amsmath,amssymb,amsthm}

\usepackage{graphicx}

\newtheorem{thm}{Theorem}
\begin{document}

\title{Approximating regions of attraction of a sparse polynomial differential system\footnote{The research was partly funded by R\'eseaux de Transport d'Electricit\'e (RTE).}} 

\author{Matteo Tacchi, Carmen Cardozo, Didier Henrion, Jean Bernard  Lasserre}

\maketitle

\begin{center}
\textit{\textbf{Matteo Tacchi}: R\'eseaux de Transport d'Electricit\'e, Immeuble WINDOW - 7C, Place du D\^ome, 92073 Paris La D\'efense, France;
LAAS-CNRS, 7 avenue du colonel Roche, 31031 Toulouse, France. (e-mail: {\tt tacchi@laas.fr})}

\textit{\textbf{Carmen Cardozo}: R\'eseaux de Transport d'Electricit\'e, Immeuble WINDOW - 7C, Place du D\^ome, 92073 Paris La D\'efense, France. (e-mail: {\tt carmen.cardozo@rte-france.com})}

\textit{\textbf{Didier Henrion}: LAAS-CNRS, 7 avenue du colonel Roche, 31031 Toulouse, France; Faculty of Electrical Engineering, Czech Technical University in Prague, Technick\'a 2, 16627 Prague, Czechia. (e-mail: {\tt henrion@laas.fr})}

\textit{\textbf{Jean Bernard Lasserre}: LAAS-CNRS, 7 avenue du colonel Roche, 31031 Toulouse, France; Toulouse Mathematics Institute, 118 route de Narbonne, 31062 Toulouse, France. (e-mail: {\tt lasserre@laas.fr})}
\end{center}

\begin{abstract}                
Motivated by stability analysis of large scale power systems, we describe how the Lasserre (moment - sums of squares, SOS) hierarchy can be used to generate outer approximations of the region of attraction (ROA) of sparse polynomial differential systems, at the price of solving linear matrix inequalities (LMI) of increasing size. We identify specific sparsity structures for which we can provide numerically certified outer approximations of the region of attraction in high dimension. For this purpose, we combine previous results on non-sparse ROA approximations with sparse semi-algebraic set volume computation.
\end{abstract}



\section{Introduction}

This paper describes a computational technique for generating outer approximations of finite time regions of attraction (ROA) of sparse polynomial ordinary differential equations (ODEs), with the purpose of assessing the stability of large scale power systems. These outer approximations contain all initial conditions for which the dynamical systems can operate safely. Indeed, power networks are usually modeled by an interconnection of weakly coupled nodes, while the dynamic behaviour of the system is mainly driven by generators, which are modeled by (closed-loop controlled) ordinary differential equations.

Most of the technical literature on stability analysis for power networks focuses on the construction of Lyapunov functions computed by nonconvex optimization, and more specifically a bilinear variant of polynomial sums of squares (SOS) optimization, as in e.g. \cite{AMP13, TMA+18, ISX+18}. An inner approximation of the infinite time ROA is then modeled as a sublevel set of the Lyapunov function, and various heuristics are used to enlarge this sublevel set as much as possible, see e.g. \cite{C11} and references therein. It can be enforced that the Lyapunov functions have the same sparsity structure as the system to be analyzed, see e.g. \cite{Z19} and references therein, but to our knowldge, it was never applied to ROA approximation. The work of \cite{KA15, KA17} is a first step towards the application of Lyapunov techniques to ROA estimation for interconnected systems.

In \cite{HK14} the authors derive an infinite-dimensional linear programming approach to finite time ROA computation, with a primal problem on measures and a dual problem on continuous functions. Computationally speaking, measures (resp. continuous functions) are discretized into moments (resp. polynomial sums of squares, SOS) of increasing degrees, resulting in a hierarchy of finite dimensional convex optimization problems, usually semidefinite programming (SDP) problems or linear matrix inequalities (LMI). This is an application of the so-called moment-SOS or Lasserre hierarchy, a mathematical technology that can be used to solve a variety of problems in applied mathematics and engineering, see \cite{L10, H13, L15}. In the context of ROA approximation, the hierarchy generates a family of outer approximations that become tighter as the degree increases. Non-sparse approximations converge in volume to the ROA when the degree tends to infinity, see \cite[Theorem 6]{HK14}. Inner approximations of the ROA can be constructed as well, see \cite{KHJ13}. The convergence proof relies on previous work on the application of the moment-SOS hierarchy for computing the volume of semi-algebraic sets, see \cite{HLS09}.

The main contribution of the current paper is to identify a sparsity structure that allows us to apply the moment-SOS hierarchy for sparse ROA approximation. We construct a hierarchy of outer approximations of increasing degree, though the sparsity is introduced at the price of the convergence proof that no longer holds. For this, we rely heavily on \cite{TWL+19} which focuses on the approximation of the volume of a sparse semi-algebraic set.

In the context of power systems stability analysis, our paper can be seen an extension to large-scale systems of results of \cite{JMT+19, OCP+19, OTH19}. It can be interpreted as well as a finite time dual approach to the standard Lyapunov approach of \cite{C11, AMP13, KA17, TMA+18, Z19}. We prefer however to see the Lyapunov approach as a dual to an infinite time occupation measure approach, in the sense that Lyapunov functions are obtained as a (dual Lagrangian) certificate of a property (stability) of the system's trajectories (modeled by occupation measures in a primal problem). The advantage of considering finite time ROA instead of standard Lyapunov ROA is the linearity of its characterization (which leads to solving convex LMIs instead of nonconvex bilinear matrix inequalities as in the Lyapunov framework), as well as the proof of convergence in volume (in the non-sparse case).

Section \ref{sec:pb} is dedicated to the problem statement, while in section \ref{sec:hiera} we introduce the framework that we use for our computations. Section \ref{sec:results} presents our main results, and our numerical experiments are gathered in section \ref{sec:num}. Finally, we give our conclusions and perspectives in section \ref{sec:ccl}.

\section{Problem Statement} \label{sec:pb}

Let $N \in \mathbb{N}$. We consider the following system of sparsely coupled polynomial ODEs:
\begin{equation} \label{eq:pode}
\begin{array}{llr}
\dot{x}_i = f_i(x_i,x_{i+1}) \quad & x_i \in X_i & \quad i = 1 \ldots N-1 \\
\dot{x}_N = f_N(x_{N-1},x_N) \quad & x_N \in X_N &
\end{array}
\end{equation}
where $X_1,\dots,X_N$ are finite dimensional compact semialgebraic sets and $f_1,\ldots,f_N$ are polynomial maps. We define
$X := X_1 \times \ldots \times X_N$, $f : = (f_1,\ldots,f_N)$
and $n := \dim X$. Note that $X$ is also a compact semialgebraic set.

Given a finite time horizon $T>0$ and a compact semialgebraic target set $X^T := X^T_1\times\ldots\times X^T_N$,
we aim at computing outer approximations of the finite time region of attraction (ROA), defined as
\begin{equation} \label{def:ais}
X^0(T,X^T) := \left\{x^0 \in \mathbb{R}^n : \begin{array}{l}
 x(t|x^0) \in X \ \forall t \in [0,T]\\
x(T|x^0) \in X^T
\end{array} \right\}
\end{equation}
where $x(t|x^0)$ denotes the value at time $t$ of the unique solution to \eqref{eq:pode} with initial condition $x^0$. In the following, the dependance of $X^0$ in $T$ and $X^T$ will be implicit.
\section{Framework} \label{sec:hiera}

\subsection{Infinite dimensional optimization}

It is shown in \cite{HK14} that $X^0$ can be obtained as the support of the measure $\mu^0$ solution to the infinite dimensional linear programming (LP) problem
\begin{equation} \label{lp:pdense}
\begin{array}{rcl}
p^\ast := &  \max & {\mu^0(X)} \\
& \mathrm{s.t.} & \lambda^n - \mu^0 \in M(X)_+ \\
& & \partial_t \mu + \text{div}(f \ \mu) + \delta_T\mu^T = \delta_0\mu^0
\end{array}
\end{equation}
where the unknowns are measures $\mu^0 \in M(X)_+$, $\mu \in M([0,T]\times X)_+$, $\mu^T \in M(X^T)_+$
and $M(X)_+$ denotes the cone of Borel measures on $X$, $\lambda^n$ is the $n$ dimensional Lebesgue measure such that $\int \psi \ d\lambda^n = \int \psi(x) \ dx$ for any measurable $\psi : \mathbb{R}^n \to \mathbb{R}$ and $\delta_t$ denotes the Dirac measure in $t$ such that $\int \phi \ \delta_t = \phi(t)$ for any continuous $\phi : [0,T] \to \mathbb{R}$.

The last constraint is the so-called Liouville transport partial differential equation (PDE) whose characteristics are exactly the trajectories associated to solutions of the ODE \eqref{eq:pode}. \cite{HK14} proved that the solution of \eqref{lp:pdense} is given by
\begin{equation} \label{sol:dense}
\begin{array}{ll}
\mu^0(A) & = \ \lambda(A \cap X^0)\\
\mu(I\times A) & = \ \int_{I\cap [0,T]} \int_X \chi_{A}(x(t|x^0)) \ \mu^0(dx^0) \ dt\\
\mu^T(A) & = \ \int \chi_A(x(T|x^0)) \ \mu^0(dx^0)
\end{array}
\end{equation}
for any Borel sets $A \subset X$ and $I \subset [0,T]$, where $\chi$ denotes the boolean indicator function. These measures are respectively called the \textit{initial measure}, the \textit{occupation measure} and the \textit{terminal measure}.

Problem \eqref{lp:pdense} has a dual formulation on functions that can be formulated as follows:
\begin{equation} \label{lp:ddense}
\begin{array}{rcl}
d^\ast := & \inf & \int_X w \ d\lambda^n \\
& \mathrm{s.t.} & w \geq v(0,\cdot) + 1\\
&& \partial_t v  + \nabla v  \cdot f  \leq 0\\
&& v(\cdot,T) \geq 0  \:\mathrm{on}\:X^T 
\end{array}
\end{equation}
where the unknowns are functions $v \in C^1([0,T]\times X)$, $w \in C^0(X)$
and $C^k(X)$ denotes the vector space of continuous and $k$ times continuously differentiable functions on $X$, and $C^k(X)_+ := \{\psi \in C^k(X) : \forall x \in X , \psi(x) \geq 0\}$ is its cone of nonnegative elements.

It was shown in \cite{HK14} that any $v$ feasible for \eqref{lp:ddense} is such that $X^0_v := \{x^0 \in X : v(0,x^0) \geq 0\} \supset X^0$ using the fact that $v$ decreases along trajectories (similarly to Lyapunov functions). In addition to that, Henrion and Korda also proved that $p^\ast = d^\ast$ (this is the strong duality property), from which one can deduce the existence of a sequence of feasible polynomials $(v_k,w_k)_{k\in\mathbb{N}}$ of increasing degrees such that $X^0_{v_k}$ converges to $X^0$ in volume, \textit{i.e.}
$$ \lambda^n\left(X^0_{v_k} \setminus X^0\right) \rightarrow 0 \:\:\mathrm{when}\:\:k\to\infty.$$
Such a sequence of polynomials is computed through convex optimization using the Moment-SOS hierarchy.

\subsection{The Moment-SOS hierarchy}

The Moment-SOS hierarchy is a primal-dual hierarchy of convex programs that grow in size and whose solutions give certified approximations to the solutions of infinite dimensional LPs on measures and functions. For the sake of simplicity and since we do not resort to the Moment side of the hierarchy, we will only present the SOS part.

A polynomial $\sigma \in \mathbb{R}[x]$ is called a \textit{sum of squares (SOS)} if it can be written $\sigma = p_1^2 + \ldots + p_k^2$ for some $k \in \mathbb{N}$, $p_i \in \mathbb{R}[x]$. Sums of squares are related to positive and nonnegative polynomials through Putinar's Positivstellensatz (P-satz):

\begin{thm}[\cite{P93}, Theorem 1.3]\label{thm:putinar}
Given polynomials $g_1,\dots,g_N \in \mathbb{R}[x]$ such that $g_1(x) = R^2 - |x|^2$ for some $R > 0$, let $X := \{x \in \mathbb{R}^n : g_1(x) \geq 0 , \dots , g_N(x) \geq 0\}$, $\Sigma(X) := \{\sigma_0 + \sigma_1g_1 + \ldots + \sigma_Ng_N, \:\sigma_0, \sigma_1, \ldots \sigma_N \text{ SOS}\}$, $P(X) := \{p \in \mathbb{R}[x] : p \geq 0 \:\mathrm{on}\:X\}$ and $P^\ast(X) := \{p \in \mathbb{R}[x] : p  > 0 \:\mathrm{on}\:X\}$. Then, one has
\begin{equation}
P^\ast(X) \subset \Sigma(X) \subset P(X).
\end{equation}
\end{thm}

Thus, using the Stone-Weierstrass Theorem (stating that any continuous function can be approximated with polynomials on a compact set) and Theorem \ref{thm:putinar}, one can look for feasible plans to \eqref{lp:ddense} by replacing all inequality constraints with SOS constraints. This way, one obtains a hierarchy of SOS problems indexed with the degree of the unknown polynomials. The infinite dimensional LP problem is thus turned into a hierarchy of finite dimensional SDP problems (\textit{i.e.} convex optimization problems with LMI constraints), using a last result:

\begin{thm}[see e.g. \cite{L10}, Proposition 2.1]\label{thm:sos}
The constraint that a polynomial is SOS is a linear matrix inequality (LMI).
\end{thm}

This work was extended to inner approximations of the ROA in \cite{KHJ13}, and to inner approximations of the maximum positively invariant set in \cite{OTH19}. It has the advantage to reduce the estimation of finite time ROA to a hierarchy of convex optimization programs, while Lyapunov-based methods for ROA estimation rely on nonconvex bilinear matrix inequalities, see e.g. \cite{ISX+18}.

The issue that one often encounters is that the moment-SOS hierarchy resorts to semidefinite programming (SDP) which does not scale well (\cite{OCP+19} pushed to dimension 5 state space). In order to tackle higher dimensional problems, one has to exploit additional properties such as sparsity or symmetries.



\section{Main Results} \label{sec:results}

\subsection{A sparse infinite dimensional program}

With the application to electrical power systems in mind, we focus on exploiting the network-like structure in our computations. A power network model has the particularity that not all the variables directly interact in the equations. Especially, nodes that are geographically far from each other are not connected together in the dynamics of the system. This corresponds to a sparse structure, which motivates this work.

Following the inspiration given by both \cite{TWL+19} and \cite{Z19}, we derive an LP problem that can be split into small dimensional subproblems, and thus is a lot more scalable than the dense formulation \eqref{lp:ddense}. To that end, we introduce the number of cliques $K := N-1$ as well as the compact semialgebraic sets $Y_i := X_i \times X_{i+1}$, $Y^T_i := X^T_i \times X^T_{i+1}$ and $n_i := \dim Y_i$, $i=1,\ldots,K$. Then, one can write the following LP
on functions:

%

\begin{subequations} \label{lp:dsparse}
\begin{eqnarray}
d_s^\ast := & \inf & \int_{Y_1} w_1 \ d\lambda^{n_1} + \ldots + \int_{Y_K} w_K \ d\lambda^{n_K} \label{obj:dsparse}\\
&& w_j \geq v_{j1}(0,\cdot) + v_{j2}(0,\cdot) + 1\label{cst:inic}\\
&& w_K \geq v_K(0,\cdot) + 1 \label{cst:inif}\\
&& v_{j1}(T,\cdot) + v_{j2}(T,\cdot) \geq 0 \:\mathrm{on}\:Y^T_j\label{cst:finic} \\
&& v_K(T,\cdot) \geq 0 \:\mathrm{on}\:Y^T_K\label{cst:finif} \\
&& u_j + \partial_{x_{j+1}}v_{j2}\cdot f_{j+1} \leq 0 \label{cst:trns} \\
&& \partial_t v_{j1} + \partial_t v_{j2} + \partial_{x_j}v_{j1}\cdot f_j \leq u_j \label{cst:lyapc} \\
&&  \partial_t v_K + \partial_{\left(\substack{x_K\\x_N}\right)}v_K \cdot \left(\substack{f_K\\f_N}\right) \leq 0 \label{cst:lyapf}
\end{eqnarray}
\end{subequations}
where the unknowns are functions $w_i \in C(Y_i)_+$, $u_j \in C^0([0,T]\times X_{j+1})$, $v_{j1} \in C^1([0,T]\times X_j)$, $v_{j2} \in C^1([0,T]\times X_{j+1})$, $v_K \in C^1([0,T]\times Y_K)$.
Here the idea is to split the decision variables $v$ and $w$ of problem \eqref{lp:ddense} and distribute them along the components of our sparse system. The decision variables $u_j$ are added to take into account the interconnexion between the components. By doing so, we end up with inequality constraints involving only the variables of one of the considered subsystems at a time, which drastically reduces the dimension of the decision space in the SOS hierarchy.

Our main result is the numerical certification that can be stated as follows:

\begin{thm} \label{thm:outer}
Let $(w_j,w_K,u_j,v_{j1},v_{j2},v_K)_{j=1,\dots,K-1}$ be a feasible plan for problem \eqref{lp:dsparse}, and consider the set
\begin{equation}
X^0_{(v_j)_j} := \left\{x \in \mathbb{R}^n : \begin{array}{l} (v_{j1}(0,x_j) + v_{j2}(0,x_{j+1}))_j \geq 0 \\ v_K(0,x_K,x_N) \geq 0\end{array}\right\}.
\end{equation}
Then, one has $X^0 \subset X^0_{(v_j)_j}$.
\end{thm}
\begin{proof}
Let $x^0 \in X^0$. Then, by definition, $x(T|x^0) \in X^T$, and according to constraint \eqref{cst:finic} one has
$$v_{j1}(T,x_j(T|x^0)) + v_{j2}(T,x_{j+1}(T|x^0)) \geq 0.$$
Moreover we know that
\begin{eqnarray*}
v_{j1}(T,x_j(T|x^0)) - v_{j1}(0,x_j^0) &=& \hspace*{-1.3pt}\int_0^T \frac{d}{dt}(v_{j1}(t,x_j(t|x^0))) \ dt \\
&=& \int_0^T \partial_tv_{j1}(t,x_j(t|x^0)) + \\
&& \hspace*{-9em} \partial_{x_j}v_{j1}(t,x_j(t|x^0))\cdot f_j(x_j(t|x^0),x_{j+1}(t|x^0)) \ dt \\
& \stackrel{\eqref{cst:lyapc}}{\leq} & \int_0^T u_j(t,x_{j+1}(t|x^0)) - \\
&& \partial_tv_{j2}(t,x_{j+1}(t|x^0)) \ dt.
\end{eqnarray*}
The same reasoning on $v_{j2}$ yields
$$ \begin{array}{l}
v_{j2}(T,x_{j+1}T|x^0)) - v_{j2}(0,x_{j+1}^0) \stackrel{\eqref{cst:trns}}{\leq} \\
\hspace*{5em}\int_0^T \partial_tv_{j2}(t,x_{j+1}(t|x^0)) - u_j(t,x_{j+1}(t|x^0)) \ dt.
\end{array} $$
Finally, adding both inequalities, one obtains
\begin{eqnarray*}
0 &\leq & v_{j1}(T,x_j(T|x^0)) + v_{j2}(T,x_{j+1}(T|x^0)) \\
&\leq & v_{j1}(0,x_j^0) + v_{j2}(0,x_{j+1}^0).
\end{eqnarray*}

Since $v_K$ is nonnegative at time $T$ in $Y^T_K$ in virtue of \eqref{cst:finif}, and decreasing along trajectories in virtue of \eqref{cst:lyapf}, the last required inequality is also satisfied. Thus, $x^0 \in X^0_{(v_j)_j}$.
\end{proof}

With this formulation, we design a method to compute sparse outer approximations of the ROA, using only convex semidefinite programming, while all existing methods resort only to nonconvex optimization, namely bilinear matrix inequalities. However, the constraint that the approximation should be sparse is a significant restriction that prevents us from proving convergence to the actual ROA. Indeed, with the following elementary exemple we show that in the case of sparse dynamics and target set, the ROA has no reason to be sparse.

\subsection{Is the ROA sparse ?} \label{sec:nosparse}

Consider the simple case where $N = 3$ and the dynamics are:

\begin{subequations}\label{ode:bicyl}
\begin{eqnarray}
\dot{x}_1 &=& (x_1^2 + x_2^2 - 0.25)x_1 \\
\dot{x}_2 &=& (x_2^2 + x_3^2 - 0.25)x_2 \\
\dot{x}_3 &=& (x_2^2 + x_3^2 - 0.25)x_3.
\end{eqnarray}
\end{subequations}

Here, it is clear that the bicylinder $B := \{x\in\mathbb{R}^3 : x_1^2 + x_2^2 \leq 0.25, \  x_2^2 + x_3^2 \leq 0.25\}$ is contained in the infinite time ROA of the equilibrium point $0$.

However, this ROA is strictly larger than our sparsely defined $B$, and it intricates all variables, which means that it cannot be sparsely described.

To illustrate this fact, we plotted the evolution of $x_1(t|x^0)$ with different initial conditions $x^0$ outside the $B$ (see Figures \ref{fig:stable} and \ref{fig:unstable}). In the three cases, $(x_2^0,x_3^0)$ is in the disk of radius $0.5$ such that $x_2(t|x^0)$ and $x_3(t|x^0)$ go to $0$ quickly.

\begin{figure}[h]
\begin{center}
\includegraphics[scale=0.5]{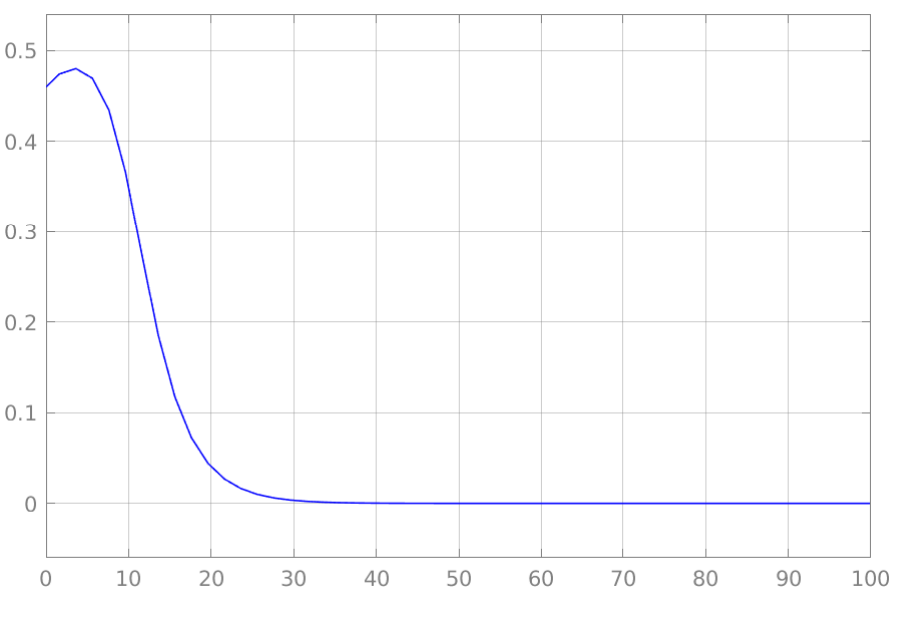}
\caption{$x_1(t|x^0)$ with $x_1^0 = 0.46$, $x_2^0 = x_3^0 = 0.25$.}
\end{center}
\label{fig:stable}
\end{figure}

\begin{figure}
\begin{center}
\hspace*{2.5em}
\includegraphics[trim={0 0 11cm 0},clip,scale=0.5]{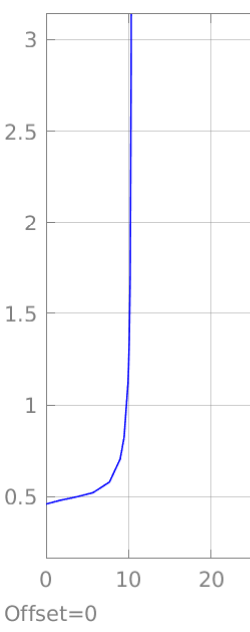}
\hspace*{2.5em}
\includegraphics[trim={0 0 11cm 0},clip,scale=0.5]{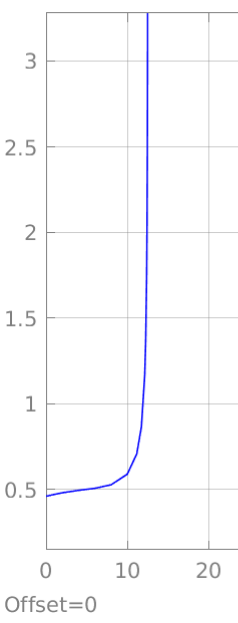}
\hspace*{2.5em}
\caption{$x_1(t|x^0)$ with $x_1^0 = 0.46$, $x_2^0 = 0.26$, $x_3^0 = 0.25$ (left) and $x_1^0 = 0.46$, $x_2^0 = 0.25$, $x_3^0 = 0.3$ (right).}
\label{fig:unstable}
\end{center}
\end{figure}

However, depending on {both} $x_2^0$ {and} $x_3^0$, the trajectory of $x_1(t|x^0)$ is either stable (with quick convergence to $0$) or unstable (with finite time explosion).

This example highlights the non-sparsity of the infinite time ROA. The same observation carries over for any finite time ROA (say for $T=100$, $X^T=[-0.1,0.1]^3$) which is very close to the infinite time ROA.

From this we can deduce that exploiting sparsity prevents us to ensure the convergence of our ROA estimations towards the actual ROA, the former being sparsely defined while the latter is not. However, we can still obtain good outer approximations of the ROA using this technique. The advantages that one gains while giving up convergence are twofold :
\begin{itemize}
\item The computational time is drastically reduced for systems that were tractable using the converging dense framework.
\item This framework allows to handle systems that are intractable with the standard dense framework, as shown experimentally below.
\end{itemize}

\section{Numerics} \label{sec:num}

We tested our formulation \eqref{lp:dsparse} on two numerical examples: the first one is the example that we mentioned in section \ref{sec:nosparse}, and the second one is a dimension $20$ chain constituted by $10$ interconnected Van der Pol oscillators.

\subsection{Reducing computational time: a toy example}

To check that our sparse method is relevant, we implemented it on system \eqref{ode:bicyl} and compared its performances to those of the dense formulation of \cite{HK14}, with SOS polynomials of degrees $8$ (Figure \ref{fig:bicyld8}) and $10$ (Figure \ref{fig:bicyld10}), with state constraint set $X = [-1,1]^3$, time horizon $T = 100$ and target set $X^T = [-0.1,0.1]^3$.

On these figures we also plot the bicylinder (that should be inside the infinite time ROA and the finite time ROA $X^0$ for $T$ large enough and $X^T$ small enough).

We gathered the computational times in Table \ref{tbl:cpu}.

\begin{table}[h]
\begin{center}
\begin{tabular}{|c|c|c|}
\hline
degree & dense & sparse \\
\hline
$4$ & $4$ & $4$ \\
\hline
$6$ & $24$ & $10$ \\
\hline
$8$ & $334$ & $83$ \\
\hline
$10$ & $5542$ & $440$\\
\hline
$12$ & - & $1865$\\
\hline
\end{tabular}
\caption{Computation times (in seconds) for the dense and sparse formulations.}
\label{tbl:cpu}
\end{center}
\end{table}

The first thing one can note is the important gain in computational time: our sparse formulation is by far less costly than the standard dense formulation, and the gap increases with the degree (at degree $10$ the sparse formulation is more than $10$ times faster).

\begin{figure}[h]
\begin{center}
\includegraphics[scale=0.3]{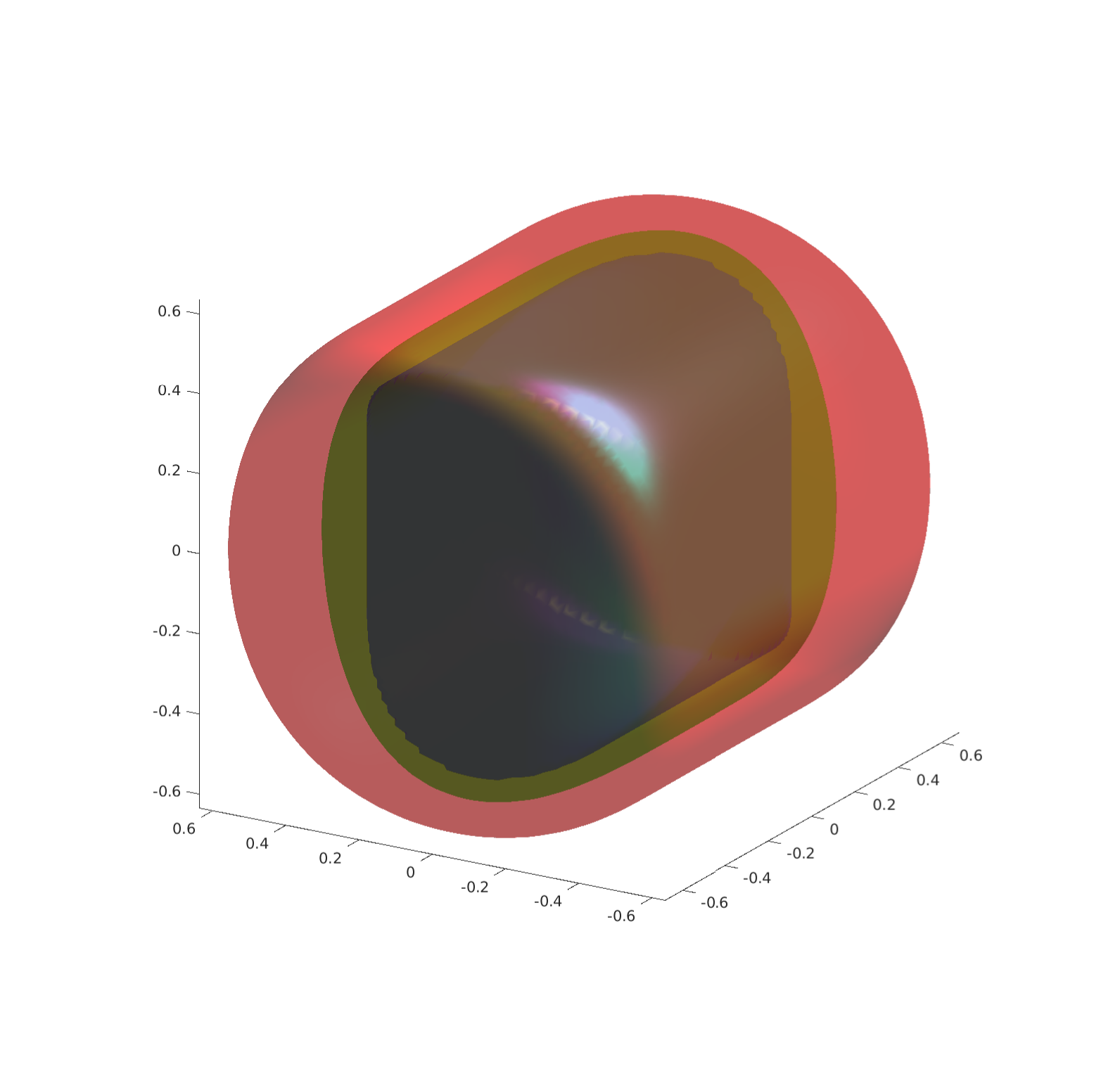}
\caption{Degree 8 sparse (red) and dense (green) ROA approximations and the bicylinder (brown).}
\label{fig:bicyld8}
\end{center}
\end{figure}

\begin{figure}
\begin{center}
\includegraphics[scale=0.3]{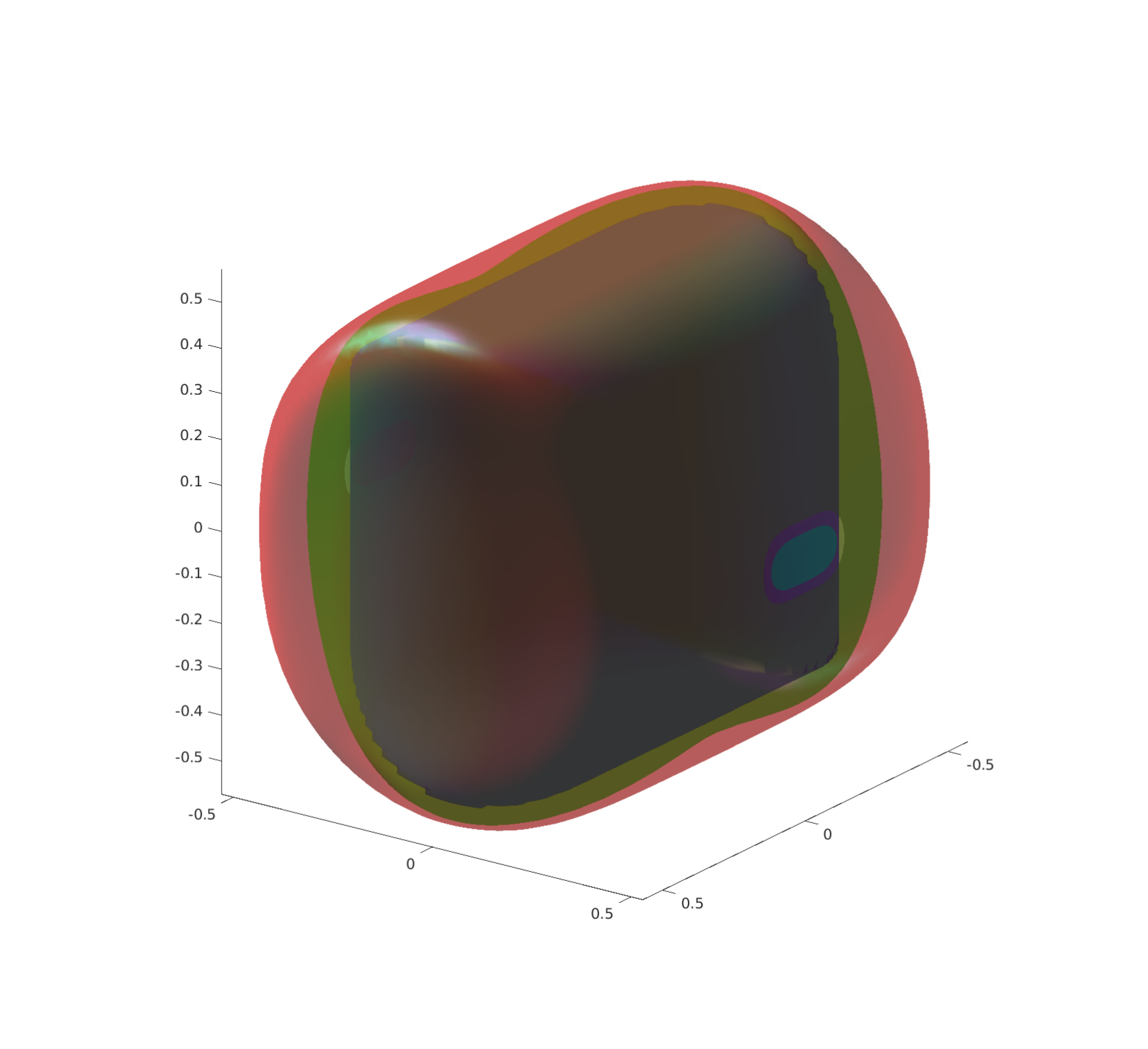}
\caption{Degree 10 sparse (red) and dense (green) ROA approximations and the bicylinder (brown).}
\label{fig:bicyld10}
\end{center}
\end{figure}

Second, one can see that while at degree $8$ the sparse approximation is less tight than the dense one, at degree $10$ the dense approximation outperforms the sparse one around $x=0$ (blue-green spot on the side of the surface), and more generally both approximations are close one to another.

\subsection{High dimension: a chain of Van der Pol oscillators}

To test our method on large scale systems, we take the same example as in \cite{KA15}, but adapted to our first sparsity pattern: we consider a chain of Van der Pol oscillators linked with random couplings. The general framework is as follows:
\begin{subequations} \label{ode:vdp}
\begin{eqnarray}
\dot{y}_i &=& -2z_i \\
\dot{z}_j &=& 0.8 y_j + 10 (1.2^2y_j^2-0.21) z_j + \epsilon_j z_{j+1} y_j \\
\dot{z}_K &=& 0.8 y_K + 10 (1.2^2y_N^2-0.21) z_K
\end{eqnarray}
\end{subequations}
with $i=1,\dots,K$ and $j=1,\dots,K-1$. This corresponds to our sparse polynomial ODE \eqref{eq:pode} with $n_i = 2$ and $x_i = (y_i,z_i)$ for $i=1,\dots,K-1$, $n_K = n_{K+1} = 1$ and $x_K = y_K$ and $x_{K+1} = z_K$ (thus $n = 2K$). One can notice that the sparse structure is even more specific than stated in our general framework since $f_j(x_j,x_{j+1}) = \left(\substack{g_j(z_j) \\ h_j(x_j,z_{j+1}) }\right)$ for $j=1,\dots,K-1$.

Here $\epsilon_j$ is a random variable that follows the uniform law on $[-0.5,0.5]$, modelling a weak interaction between the oscillators. For reporting our results, we let $K=10$, $X = [-1,1]^{20}$, $T = 30$ and $X^T = [-0.1,0.1]^{20}$ and we use a particular sample $\epsilon$.
We report on degree $12$ certificates, which takes approximately 23', among which 11'35" for declaring the decision variables with the YALMIP interface, 10'46" for solving the SDP problem with MOSEK and 41" for plotting the results with Matlab.

For $j = 1,\dots,K-1$ we plot the sets
$$X^0_j := \{x_j \in X_j : v_{j1}(0,x_j) + v_{j2}(0,0) \geq 0 \} $$
which correspond to $T=30$, $X^T=[-0.1,0.1]^2$ for the $j$-th Van der Pol oscillator with perturbation $\epsilon_j z_{j+1} y_j$ where $z_{j+1}$ is a trajectory from the $(j+1)$-th Van der Pol oscillator, starting in $0$ at $t=0$, see  Figure \ref{fig:cliques}.

\begin{figure}
\hspace*{-1em}\includegraphics[scale=0.22]{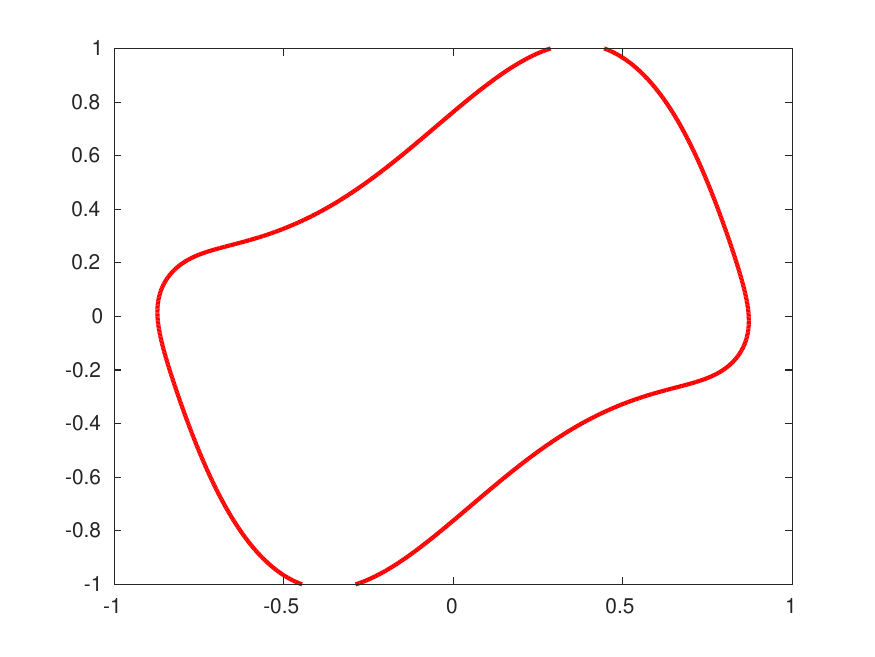}
\hspace*{-1.2em}\includegraphics[scale=0.22]{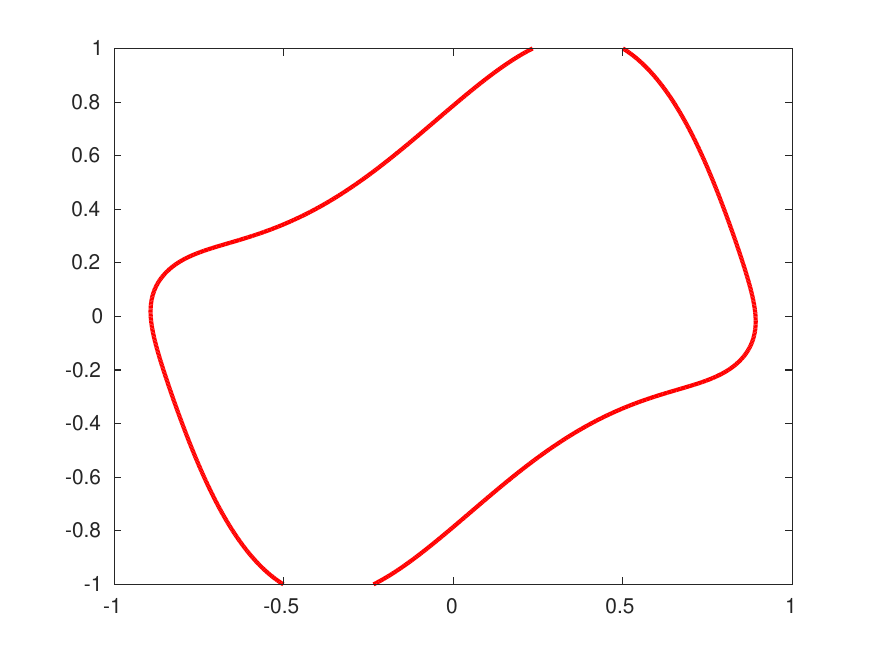}
\hspace*{-1.2em}\includegraphics[scale=0.22]{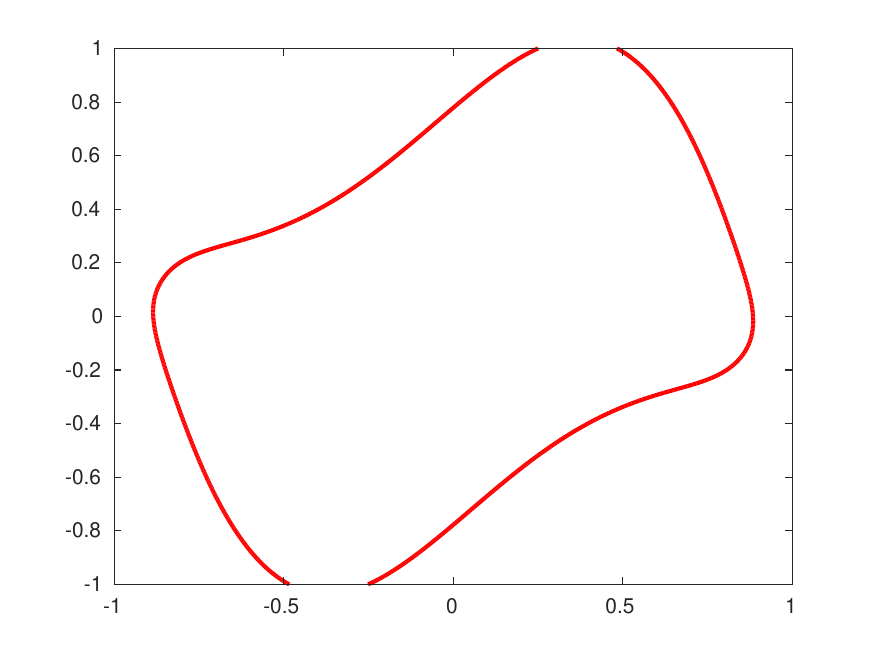}
\hspace*{-1.2em}\includegraphics[scale=0.22]{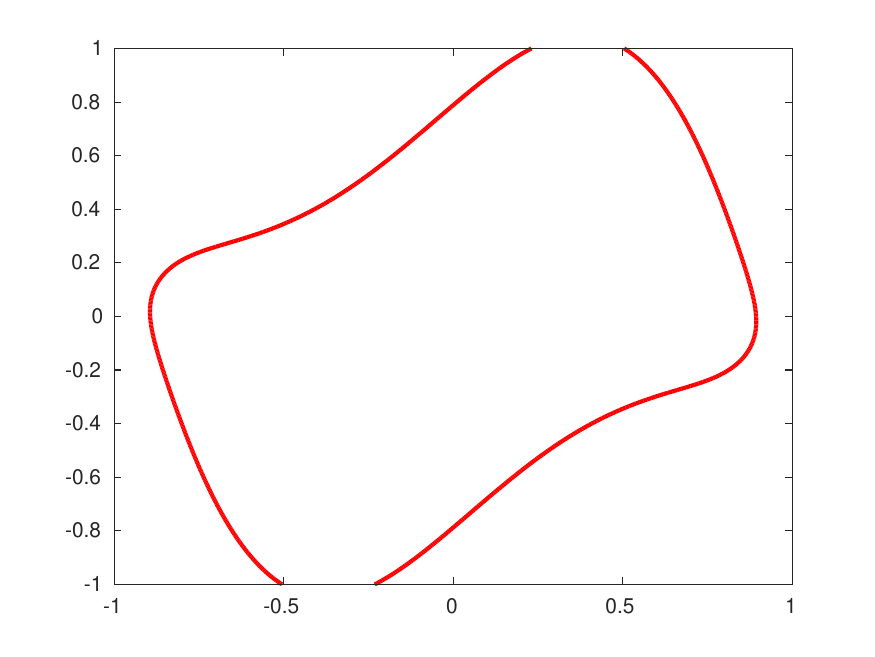}
\hspace*{-1.2em}\includegraphics[scale=0.22]{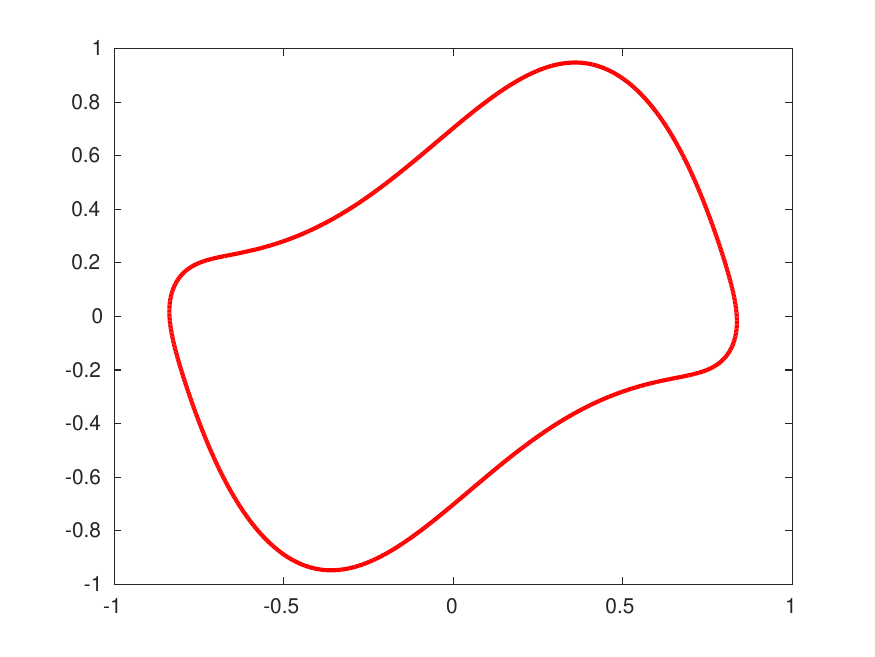}
\hspace*{-1.2em}\includegraphics[scale=0.22]{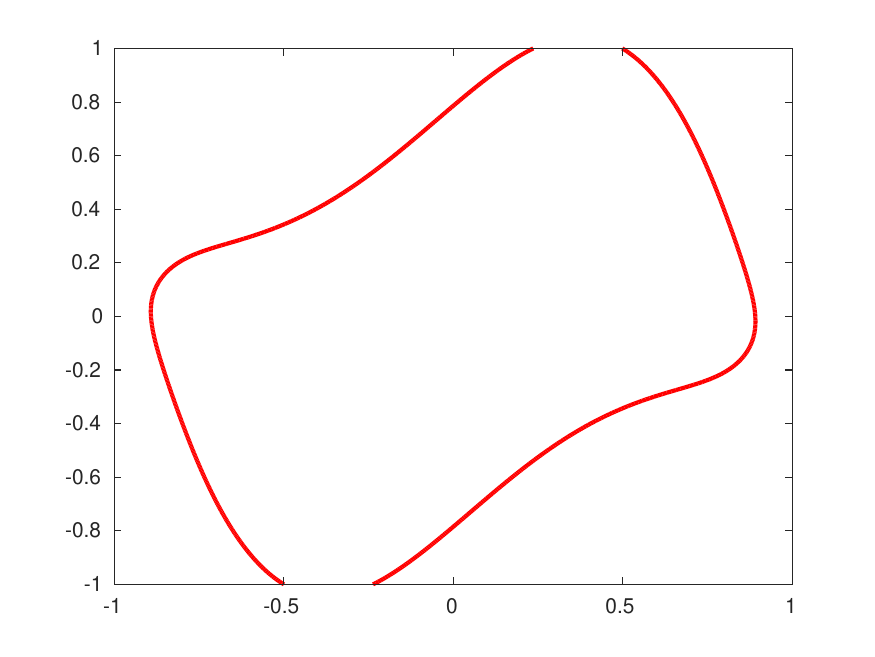}
\hspace*{-1.2em}\includegraphics[scale=0.22]{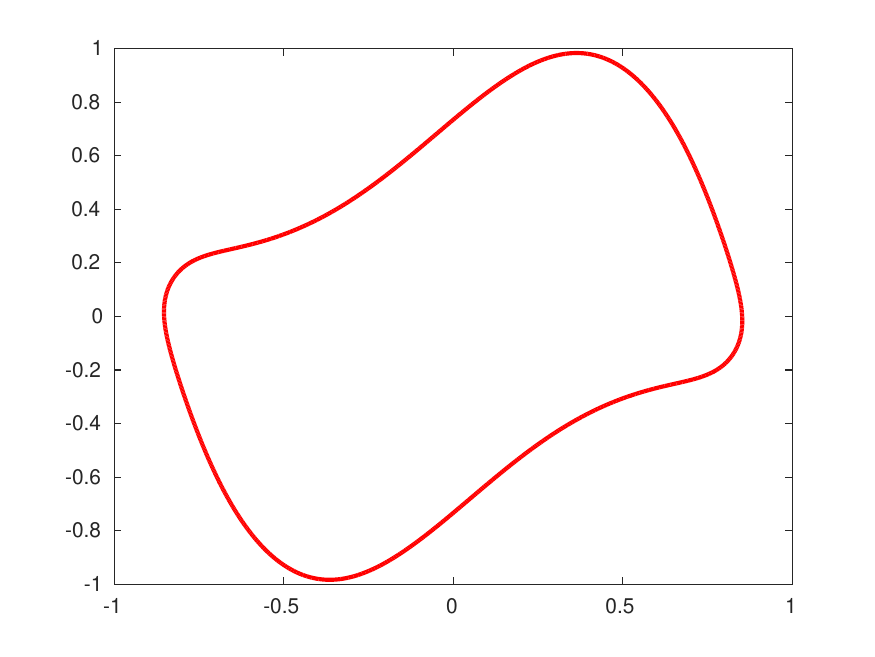}
\hspace*{-1.2em}\includegraphics[scale=0.22]{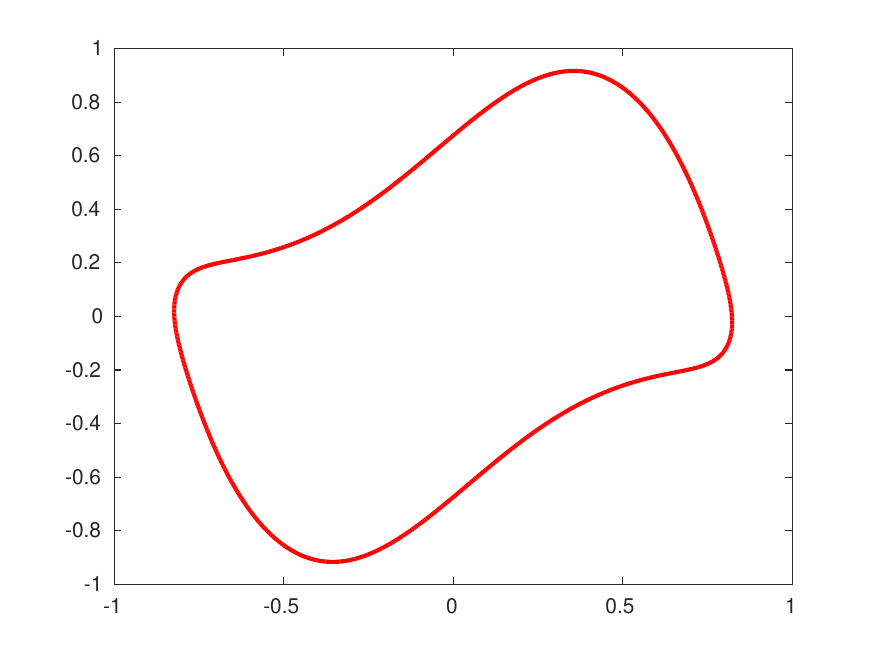}
\hspace*{-1.2em}\includegraphics[scale=0.22]{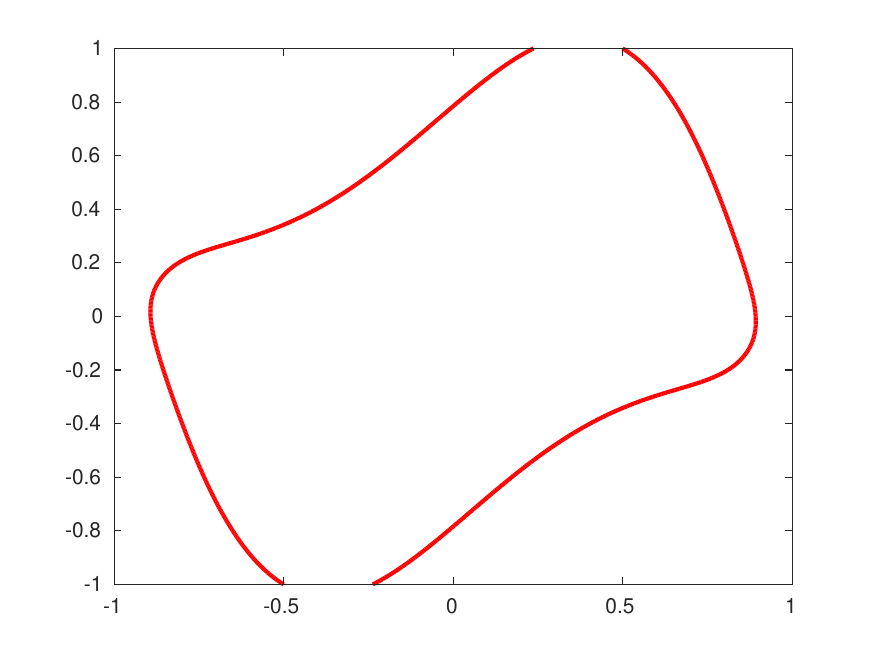}
\caption{$X^0_1$ to $X^0_9$ (from left to right and top to bottom).}
\label{fig:cliques}
\end{figure}

We also plot the set
$$ X^0_K := \{x_K \in X_K : v_K(0,x_K,x_{K+1}) \geq 0 \}$$
which corresponds to $T = 30$, $X^T=[-0.1,0.1]^2$ for the $K$-th (non perturbed) Van der Pol oscillator, see Figure \ref{fig:vdp}.

\begin{figure}
\begin{center}
\includegraphics[scale=0.3]{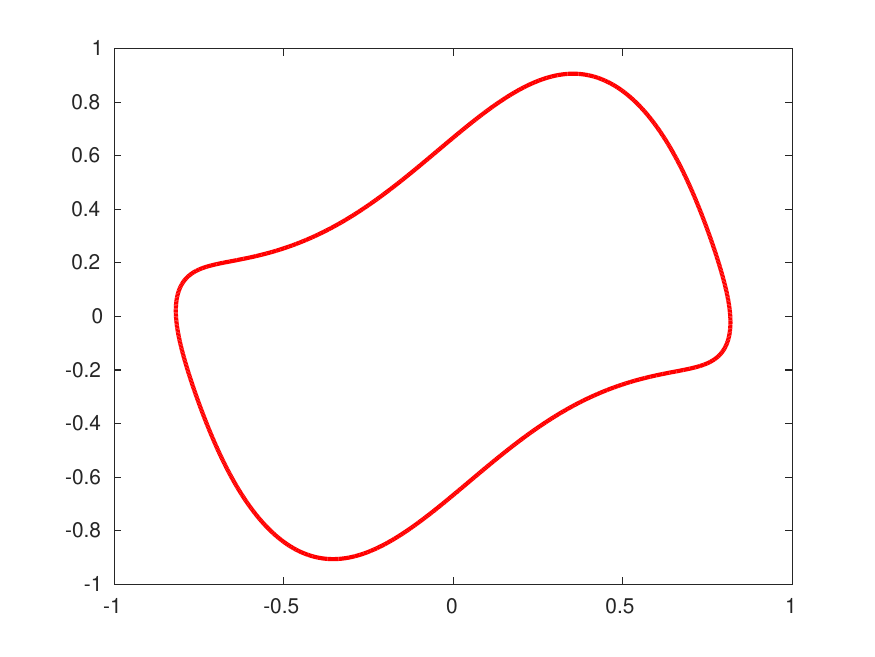}
\end{center}
\caption{$X^0_{10}$ corresponds to a regular Van der Pol oscillator.}
\label{fig:vdp}
\end{figure}

As expected considering the low magnitude of the interactions,  on Figure \ref{fig:cliques} one can identify shapes similar to the ROA of a standard Van der Pol oscillator. However, the shapes are perturbed: their respective sizes differ slightly. The standard framework of \cite{HK14} cannot be used here, due to the high dimension of the state space. It is also important to note that an important part of the computational time was spent for modelling the SDP problem, while the SDP solver was quite fast, once the decision variables were properly declared. We believe that these results are quite encouraging for future works on sparse ROA approximation.

\section{Conclusion} \label{sec:ccl}

This work is a first step towards convex computation of large scale stability regions for sparse systems. Like Lyapunov-based methods, this framework gives no convergence guarantee for the polynomial approximations when the degree tends to infinity, due to the strong sparsity constraints imposed to the SOS certificates. However, we have been able to reduce the problem of assessing stability of a large scale sparse system into a tractable convex problem. In our opinion this is a complete novelty, since previous works resulted into nonconvex bilinear problems.

This framework is valid for any chain of coupled ODEs, and it can readily be extended to other sparsity patterns, as highlighted in \cite{TWL+19}. The presentation of the results is however more complicated, which is the reason why we only presented chained ODEs in this paper.

For now, we applied it only for outer approximations of the finite time ROA, while inner approximations of the infinite time ROA and maximal positively invariant sets remain to be studied. Future work will also include the transient stability assessment of a meshed multi-machine system as in \cite{AMP13,TMA+18}, and the stability analysis of different converter grid-forming controls as in \cite{AJD18,TGA+19}.

\section*{Acknowledgements}
This work benefited from interactions with Patrick Panciatici.

\bibliography{ifacconf}             

\end{document}